\documentclass[%
aip,
amsmath,amssymb,
reprint, onecolumn %%%%% Comment onecolumn to override one column format %%%%%%%
]{revtex4-1}

\usepackage{graphicx}% Include figure files
\usepackage{dcolumn}% Align table columns on decimal point
\usepackage{bm}% bold math
\usepackage[utf8]{inputenc}
\usepackage[T1]{fontenc}
\usepackage{mathptmx}
\usepackage{etoolbox}

\usepackage{amsmath, amssymb, amsfonts, amsthm, mathrsfs}
\newtheorem{theorem}{Theorem}
\newtheorem{proposition}[theorem]{Proposition}
\newtheorem{corollory}{Corollory}
\newtheorem{example}{Example}%
\newtheorem{remark}{Remark}%
\newtheorem{definition}{Definition}%

\makeatletter
\def\@email#1#2{%
	\endgroup
	\patchcmd{\titleblock@produce}
	{\frontmatter@RRAPformat}
	{\frontmatter@RRAPformat{\produce@RRAP{*#1\href{mailto:#2}{#2}}}\frontmatter@RRAPformat}
	{}{}
}%
\makeatother
\begin{document}
	
	\preprint{AIP}
	
	\title[Concrete Quantum Channels ]{Concrete Quantum Channels and Algebraic Structure  of  Abstract Quantum Channels}
	% Force line breaks with \\
	\author{M.N.N. Namboodiri}
	\altaffiliation{Former Professor,Department of Mathematics, Cochin University of Science and Technology, Kochi, India.\\ IMRT,Thiruvananthapuram.}
	\email{mnnadri@gmail.com}
	\date{\today}
	\begin{abstract}This article analyzes the algebraic structure of the set of all quantum channels and its subset consisting of quantum channels that have Holevo representation. The regularity of these semigroups under composition of mappings is analyzed. It is also known that these sets are compact convex sets and, therefore, rich in geometry. An attempt is made to identify generalized invertible channels and also the idempotent channels. When channels are of the Holevo type, these two problems are fully studied in this article. The motivation behind this study is its applicability to the reversibility of channel transformations and recent developments in resource-destroying channels, which are idempotents. This is related to the coding-encoding problem in quantum information theory. Several examples are provided, with the main examples coming from pre-conditioner maps which assign preconditioners to matrices in numerical linear algebra. Thus, the known pre-conditioner maps are viewed as quantum channels in finite dimensions. In addition, the infinite-dimensional analogue of preconditioners is introduced and certain limit theorems are discussed; this is with an aim to analyze asymptotic methods in quantum channels analogous to problems in asymptotic linear algebra.
	\end{abstract}
	
	\maketitle
	
		\section{Introduction}\label{intro}
	Over the past few years, there has been significant progress in the construction and testing of preconditioners for Toeplitz and block Toeplitz matrices using Korovkin's classical approximation theorems through positive linear maps. These preconditioners have been found to exhibit complete positivity, leading to an abstract formulation of Korovkin-type theorems in a non-commutative context. Interestingly, these preconditioner maps can be viewed as abstract quantum channels.
	
	In this concise article, we explore this viewpoint by examining various related quantities such as the Kraus representation, channel capacity, and fidelity. We also discuss the algebraic structure of the class of quantum channels, which forms a semigroup under composition of mappings. Furthermore, we prove that the set of quantum channels in the Holevo form constitutes a sub-semigroup without an identity. We obtain a characterization of idempotents in terms of the associated stochastic matrices, which is significant as the range of idempotents remains invariant under quantum operations. Remarkably, we observe that the preconditioner maps are idempotents belonging to one of the aforementioned semigroups. In recent years, idempotent quantum channels have proved to be useful in coding-encoding problems.
	
	Completely positive linear maps between $C^{*}$-algebras have emerged as powerful tools in non-commutative analysis, quantum mechanics, quantum information theory, and linear algebra problems (see references such as \cite{ASH}, \cite{EBD}, and \cite{KMS}). As mentioned in the abstract, we focus on studying the maps $P_{U}(.)$ and $\tilde{P}_{U}$, described in the articles \cite{SSC} and \cite{KMS}, which assign good preconditioners to elements in the class $M_{n}(C)$ of $n \times n$ matrices. We observe that Lemma 2.1 on page 311 of \cite{SSC} implies that $P_{U}$ and its general versions introduced in \cite{KMS} can be viewed as abstract quantum channels. In line with this perspective, we examine these maps in the context of quantum channels (C-Q Channels) introduced and studied in references such as \cite{ASH}, \cite{JKR}, and \cite{MDC}. Inspired by the work of Alexander Holevo and David Kribs, we determine the channel fidelity, compute the Choi matrix \cite{JKR}, and analyze other relevant quantities associated with the map $P_{U}(.)$ as well as its general version, $\tilde{P}_{U}$. Additionally, we provide several concrete examples of unitary matrices $U$ that define $P_{U}$.
	
	This article is divided into five sections. Sections \ref{intro} and \ref{prelim} are devoted to the introduction and preliminaries, and a few new results of the preconditioner maps $P_{U}$ and $\tilde{P}_{U}$ concerned. A few related concepts, such as channel capacity and channel-fidelity, are considered in the next couple of sections. It is proved that the channel capacity of $P_{U}$  is infinite.
	Section \ref{sec3algebra} is mainly devoted to the algebraic structure of quantum channels, especially in the Holevo-form. Section \ref{sec4code} is devoted to the role of precondioner maps $P_{U}$ and $\tilde{P}_{U}$ as resource -destroying channels in coding-encoding problems mentioned in the introduction. Finally, in the last section \ref{sec5infi}, a few research problems for future considerations are given.\\
	As mentioned in the abstract,  idempotent quantum-channels are known  to be  useful in connection with coding-encoding problems. A possiblity of such interpretations for precoditioner map is discussed in remarks \ref{rem8} below \cite{ZXS},\cite{KZMK}.\\
	\textbf{Throughout sections from I  to IV, all spaces considered are of finite- dimentions.}
	\section{Preliminaries}\label{prelim}
	Throughout this article, we consider completely positive linear maps on $M_{n}(C)$, the set of all matrices of order$n$ with complex entries. As mentioned in the introduction, we concentrate on the maps considered in \cite{SSC} and  \cite{KMS}  for generating preconditioners in numerical linear algebra.\\
	Preconditioners are very important tools in Numerical Linear Algebra to deal with ill-conditioned 
	problems that arise naturally while considering real-world problems. In this regard,  constructing appropriate preconditioners is an equally important problem. As mentioned in the abstract, one such construction is given in [6], page 4.
	\begin{definition}{(Preconditioners)}\label{def1}
		Let $M_{n}(C)$ be the set of all  matrices of order $n$ over $C$.For a pre-assigned unitary $U \in M_{n}(C)$,let $M(U)$ denote the set of all matrices simultaniuosly diagonalised by $U$.To be more explicit,
		$$
		M(U)=\{A=U\Delta U^{*}:\Delta diagonal \}.
		$$
		A corresponding pre-conditioner $P_{U}(A)$ for a given matrix $A$  is then given by 
		$$
		P_{U}(A) =U\sigma(U^{*}AU)U^{*},
		$$
		where for a given  matrix $B$, $\sigma(B)$ denotes the matrix obtained by retaining the diagonal part of $B$ and putting zero elsewhere.
	\end{definition}
	\begin{remark}\label{rem1}
		In fact ,it can be seen that $P_{U}$ is the orthogonal projection of the $n^{2}$- dimentional Hilbert-space $M_{n}(C)$ with respect to the inner product $\langle A,B\rangle =Tr(AB^{*})$, onto the subspace $M(U)$.
	\end{remark}
	Next we recall \cite{KMS} a more general version $\tilde{P}_{U}$, which is defined as follows:
	\begin{definition}{(\cite{KMS})}\label{def2}
		Let $\{P_{k};k=1,2,...,d\}$ be a finite collection of mutually orthogonal projections such that $\Sigma_{k=1}^{d} P_{k}=I_{n}$, where $d\leq n$.Let 
		\begin{equation}\label{eq1}
			\Psi(A) =\Sigma_{k=1}^{d} P_{k}AP_{k},
		\end{equation}
		and
		\begin{equation}\label{eq2}
			\tilde{M}_{U}= \{A \in M_{n}(C):U^{*}AU= \Psi(A)\}.
		\end{equation}	
	\end{definition}	
	Recall the definition of $\tilde{P}_{U}$:
	\begin{equation}\label{eq3}
		\tilde{P}_{U}(A)=U(\Psi(U^{*}AU))U^{*},   A\in M_{n}(C).
	\end{equation}
	\begin{theorem}{\cite{KMS}}\label{thm1}
		\begin{itemize}
			%\item[{[1]}]
			With $A,B \in M_{n}(C)$
			%$$
			%P_{U}(A) =U\sigma(U^{*}AU)U^{*},
			%$$
			\item[{[1]}]
			$$
			\tilde{P}_{U}(\alpha A+\beta B) =\alpha \tilde{P}_{U}(A)+\beta \tilde{P}_{U}(B),
			$$
			\item[{[2]}] 
			$$
			\tilde{P}_{U}( A^{*}) =\tilde{P}_{U}(A)^*,
			$$
			\item[{[3]}]
			$$
			tr(\tilde{P}_{U}(A)) =Tr(A),
			$$
			\item[{[5]}]
			$$
			\|A-\tilde{P}_{U}(A)\|^{2}_F =\|A\|_{F}^{2} -\|\tilde{P}_{U}(A)\|^{2}_{F},
			$$
			where $ \|.\|_{F} $ denotes the Frobinious norm of matrices.
		\end{itemize}
	\end{theorem}
	
	It is clear that $P_{U}$ is a special case of $\tilde{P}_{U}$ if the projections $\{P_{k}:k=1,2,...d\}$ are precisely rank one pairwise orthogonal projections $\{p_{k}:k=1,2,...,n\}$.
	However, we treat them differently as and when the situation arises. \\	
	\begin{example}	\label{ex1}
		Several examples of unitaries $U$ arise in concrete situations such as discrete Fourier transforms, interpolation of functions etc. \cite{SSC},\cite{KMS}. We quote them here for the sake of completeness.

		Let $v = \left\{ {v_n } \right\}_{n \in N} \,\, \textrm{with}\,v_n  = \left(
		{v_{nj}
		} \right)_{j \le n - 1}$ be a sequence of trigonometric
		functions on an interval $I.$ Let $S = \left\{ {S_n } \right\}_{n \in N}$
		be a sequence of grids of n points on $ I,$
		namely, $S_n  = \left\{ {x_i^n ,i = 0,1,.\,.\,.\,n - 1} \right\}$.
		Let us suppose that the generalized Vandermonde matrix
		\begin{equation} \nonumber
			{ V_n  = }\left( {{ v}_{{ nj}} { (}x_i^n { )}} \right)_{{ i;j = 0}}^{{ n - 1}}
		\end{equation}
		is a unitary matrix. Then, algebra of the form $M_{{U_n}}$ is a trigonometric
		algebra if $U_n = {V_n}^{*}$ with $V_n$ a generalized trigonometric Vandermonde matrix.
		
		We quote  examples of trigonometric algebras with the following choice of the
		sequence of matrices ${U_n}$ and grid $S_n$.
		
		%\begin{eqnarray*}
		\begin{itemize}
			
			\item[{[1]}]
			$$	
			U_n = F_n=\left( {\frac{{ 1}}{{\sqrt { n} }}e^{ijx_i^n } } \right)\,\,,\,\,\,\,i,j
			=0,1,.\,.\,.\,n - 1,
			$$
			\item[{[2]}]
			$$
			S_n= \left\{ {x_i^n  = \frac{{2i\pi }}{n},i = 0,1,.\,.\,.\,n - 1} \right\} \subset
			I= \left[ { - \pi ,\pi } \right]
			$$
			\item[{[3]}]
			$$
			U_n = G_n =\left( {\sqrt {\frac{2}{n+1}} }sin(j+1)x_i^n \right)\,\,,\,\,\,\,i,j =
			0,1,.\,.\,.\,n - 1,
			$$
			\item[{[4]}]
			$$
			S_n= \left\{ {x_i^n  = \frac{{(i+1)\pi }}{n+1},i = 0,1,.\,.\,.\,n - 1} \right\}
			\subset I = \left[ {0 ,\pi } \right]
			$$
			\item[{[5]}]
			$$
			U_n= H_n =\left( \frac{{1}}{{\sqrt {n} }}\left[ {{sin(jx_i^n) + cos(jx_i^n )}}
			\right] \right)\,\,,\,\,\,\,i,j = 0,1,.\,.\,.\,n - 1.
			$$\\
			\item[{[6]}]  
			
			$$
			U_{n}  \quad arising \quad from  \quad Discete \quad Quantum \quad  Fourier \quad  Transforms \quad (DQFT).
			$$
			
		\end{itemize}
		%\end{eqnarray*}
	\end{example}
	The {\textbf complete- positivity} of $\tilde{P}_{U}$ was observed in \cite{KMS}, and this aspect was used to prove Korovkin-type theorems. As mentioned in the abstract, we highlight a different point of view regarding $\tilde{P}_{U}$, namely the \textbf{quantum -channel } property mentioned in the introduction. To be more precise, theorem \ref{thm1} above describes the conditions of an \textbf{abstract quantum channel}\cite{ASH}.\\Next, we recall the definition of complete positivity followed by a few very important representations of such maps.
	%\begin{lemma}{c-b norm of $\P_{U}$}
	%	Let $\Phi:\textit{A} \rightarrow \textit{B}$	be a quantum channel between $C^{*}$-algebras $\textit{A}  \quad   and \quad  \textit{B}$.The c-b norm $\|.\|_{cb}$ is defined as follows:
	%	$$
	%	\|\Phi\|_{cb}=sup_{n} \|\Phi \otimes I_{n}\|
	%	$$
	%\end{lemma}.
	\begin{definition}{(Complete positivity)}\label{def3}
		Let $A$ and $B$ be $C^{*}$-algebras over complex numbers and $\Phi: A \rightarrow B$ be a linear map.Then $\Phi$ is said to be completely positive if 
		\begin{equation}\label{eq4}
			\Phi^{(n)} : A\otimes M_{n}(C) \rightarrow B\otimes M_{n}(C),
		\end{equation}
		defined as 
		\begin{equation}\label{eq5}
			\Phi^{(N)} (a_{(i,j)})= (\Phi(a_{(i,j)})),
		\end{equation}
		is positivity preserving for all positive integers $n$  and 
		$(a_{(i,j)}) \in A\otimes M_{n}(C)$. 
	\end{definition}
	As mentioned in the introduction, we consider only finite dimensional cases.
	\begin{definition}\label{def4}
		Let$\mathcal{H}_{A}$ and $\mathcal{H}_{B}$ be the Hilbert spaces associated with the input and output systems of a channel, respectively. A channel in the Schrodinger picture is a linear, completely positive, trace-preserving map $\Phi$ 
		$$
		\Phi:\mathcal{T}(\mathcal{H}_{A}) \rightarrow \mathcal{T}(\mathcal{H}_{B}),
		$$ 
		and dually, a channel in the Heisenberg picture is a linear, completely positive map $\Phi^{*}$:
		$$
		\Phi^{*}:\mathcal{B}(\mathcal{H}_{A}) \rightarrow \mathcal{B}(\mathcal{H}_{B})
		$$ and dually a channel in the Heisenberg picture is a linear, completely positive map that preserves the identity. 
		Here $\mathcal{T}(H)$(respectively $\mathcal{B}(\mathcal{H}))$ denotes the set of all trace class operators on a Hilbert space $H$ (respectively bounded linear operators).
	\end{definition}
	Next, a few important representations of completely positive maps, such as \textbf{Stienspring-dilation theorem, Kraus-representation, Holevo-form}, are given below.
	\begin{theorem} \textbf{Stienspring- Dilation Theorem - finite-dimensional case}\label{thm2}
		Assume that $\mathcal{H}_{A},\mathcal{H}_{B}$ are finite-dimentional.Then for every completely positive map $\Phi^{*}:B(\mathcal{H}_{A})\rightarrow B(\mathcal{H}_{B}) $ there exists a Hilbert space $\mathcal{H}_{E}$ and an operator $V:\mathcal{H}_{A} \rightarrow\mathcal{H}_{B}\otimes \mathcal{H}_{E}$ such that
		\begin{equation}\label{eq6}
			\Phi^{*}(X)	= V^{*}(X\otimes I_{E}) ,X \in B(\mathcal{H}_{B})	
		\end{equation}	
		Dually,
		
		\begin{equation}\label{eq7}
			\Phi^{*}(S) = Tr VSV_{*},     S \in\mathcal{T}(\mathcal{H}_{A}).
		\end{equation}
		The map $\Phi^{*}$ is unital, i.e., identity preserving (the map $\Phi$ preserves trace ) if and only if the operator $V$ is isometric.
	\end{theorem}
	Next we state the well known \textbf{(Kraus- Representation)}\cite{KRB}
	\begin{theorem}\label{thm3}
		A map $\Phi^{*}$ is completely positive if and only if it can be represented in the form
		\begin{equation}\label{eq8}
			\Phi^{*}(X) = \Sigma_{k=1}^{d} V^{*}_{k}XV_{k},
		\end{equation}
		where $V_{k:}:\mathcal{H}_{A}\rightarrow \mathcal{H}_{B}$,or dually 
		\begin{equation}\label{eq9}
			\Phi(S)= \Sigma_{k=1}^{d}V_{k} SV^{*}_{k}, S\in \mathcal{T}(\mathcal{H}_{A}).
		\end{equation}
		The map $\Phi^{*}$ is unital ($\Phi$ is trace-preserving ) if and only if 
		\begin{equation}\label{eq10}
			\Sigma_{k=1}^{d} V_{k}^{\star}V_{k} = I.
		\end{equation}
	\end{theorem}
	
	This representation is not unique, but there is a representation with the minimal number of components arising from minimal Stienspring representation! Next, we compute the Kraus-representation of $P_{U}$, which is necessary for the subsequent analysis of $P_{U}$ as a quantum channel.
	\begin{theorem}\label{thm4}
		The map $\tilde{P}_{U}:M_{n}(C)\rightarrow M_{n}(C)$ has Kraus-representation given by 
		\begin{equation}\label{eq11}
			\tilde{P}_{U}(A)=\Sigma_{k=1}^{d} V_{k}^{*} AV_{k},
		\end{equation}
		where $V_{k}=U P_{k}U^{*}$ , $\Psi(A)=\Sigma_{k=1}^{d}P_{k}AP_{k}$ ,the pinching function defining $P_{U}$.
	\end{theorem}
	\begin{proof}
		The proof follows immediately by substituting the pinching $\Psi$ in the defining equation of $\tilde{P}_{U}$ and then grouping the terms appropriately.
	\end{proof}
	\begin{definition}(\textbf{Holevo - Form})\label{def5}
		A quantum-channel  $\Phi$ on $M_{n}(C)$	for which  there exist  density operators $\{S_{1},S_{2},...,S_{m}]\}$ in $M_{n}(C)$  and a resolution of identity $\{X_{1},X_{2},...,X_{m}\}$ such that
		\begin{equation}\label{eq12}
			\Phi(S) =\Sigma_{k=1}^{m} S_{k}T_{r}(SX_{k}),
		\end{equation}
		for all density operators $S$ is called the \textbf{Holevo-Form}\cite{KRB} of $\Phi$.\\
		In addition, if there exists an orthonormal basis $\{e_{1},e_{2},...,e_{n}\}$ such that 
		$S_{k}= \vert {e_{k}\rangle \langle e_{k}} \vert $ for $k=1,2,...,n$, then such a channel is called a \textbf{classical -quantum channel}(C-Q Channel).
	\end{definition}
	
	%\begin{section}{\textbf{Channel-Capacity}
	
	%\begin{definition}	
	
	%\end{definition}
	\subsection{Entanglement}
	\begin{definition}\label{def6}
		Let $\Phi $ be a quantumn channel on $M_{n}(C)\rightarrow M_{m}(C)$.Then $\Phi$ is called \textbf {entanglement breaking } if 
		for all $k \geq 1$ and states $X \in M_{n}(C)\otimes M_{m}(C)$, the output state $(id_{k} \otimes \Phi)(X)$ is separable, where $id_{k}$ is the identity map on $M_{n}(C)$.
	\end{definition}
	\begin{remark}\label{rem2}
		It is well known that $\Phi$ is entanglement breaking if and only if $\Phi$ has a following Kraus representation namely (refer Theorem 4, \cite{HOR})
		\begin{equation}\label{eq13}
			\Phi(X)=\Sigma_{j=1}^{n}V_{j}^{*} X V_{j}
		\end{equation}
		where each $V_{j}$ is of rank 1.So among $\{\tilde{P}_{U},P_{U}\}$ only $P_{U}$ is entanglement breaking.
	\end{remark}
	\section{Algebraic Structure}\label{sec3algebra}
	In this section, the algebraic structure of the collection of all quantum channels and the associated stochastic matrices \cite{KRB} are considered. One of the motivation in considering the algebraic structure is identification of \textbf{idempotent Channels}.Interestingly, recent articles  by Kamil Korzekwa, Zbigniew Puchała,etal.,\cite{KZMK} and Zi-Wen Liu etal.,\cite{ZXS} are centred arround idempotent channels, in the encoding-decoding problem
	in quantum communication theory.
	Let $QC(M_{n})$  and $HQC(M_{n})$ denote the collection of all quantum channels  and the set of all quantum channels in the Holevo-form on $M_{n}(C)$ respectively.
	\begin{theorem}\label{thm5}
		The set $QC(M_{n})$ is a semigroup with identity(i.e., monoid) and $\{\tilde{P}_{U},P_{U}: U unitary \in {M_{n}}\}$ are idempotents in $QC(M_{n})$ under composition of operators on $M_{n}(C)$.Also,the set $HQC(M_{n})$ is a sub- semigroup of $QC(M_{n})$ without identity,  containing $P_{U}$.
	\end{theorem}
	\begin{proof}
		It is straightforward to see  that $QC(M_{n})$  is a semigroup with identity.\\ The maps  $\{\tilde{P}_{U}, P_{U }:U unitary \in{M_{n}}\}$ are orthogonal projections onto the subspace $\tilde{M}_{U}$ and therefore they are idempotents.\\
		Let $\Phi_{j}:j=1,2$ be in $HQC(M_{n})$  and 
		\begin{equation}\label{eq14}
			\Phi_{1}(\rho)=\Sigma_{k} Tr(F_{k}\rho)R_{k},
		\end{equation} 
		and 
		and 
		\begin{equation}\label{eq15}
			\Phi_{2}(\rho)=\Sigma_{j} Tr(H_{J}\rho)S_{J},
		\end{equation}
		be the respective Holevo-representation.Here $\rho,R_{k},S_{j}$ are density matrices and 
		$F_{k},H_{j}$ are POVM,$j,k$.Therefore we have 
		\begin{equation}\label{eq16}
			\Phi_{1}\circ \Phi_{2}(\rho)=   \Sigma_{k}Tr[F_{k}\Sigma_{j}Tr(H_{j}\rho)S_{j}] R_{k},\\
		\end{equation}
		$$
		= \Sigma_{k}\Sigma_{j}Tr(H_{j}\rho) Tr(F_{k}S_{j})S_{j}] R_{k},\\
		$$
		$$
		=\Sigma_{k}\Sigma_{k}TrTr[F_{k}S_{j}Tr(H_{j}\rho)]R_{k} ,   \\
		$$
		$$
		=\Sigma_{j}Tr(H_{j}\rho)\Sigma_{k}Tr(F_{k}S_{j})R_{k}, \\
		$$
		Let
		$L_{J} =\Sigma_{k}Tr(F_{k}S_{j})R_{k}$.We show that $L_{j}$ is a density matrix for each $j$.Clearly,each $L_{j}$ is positive definite.Now we ahve 
		\begin{equation}\label{eq17}
			Tr(L_{j})=Tr\Sigma_{k}Tr(F_{k}S_{j})R_{k}=\Sigma_{k}Tr(F_{k}S_{j})Tr(R_{k})
		\end{equation}
		Interchanging summation, we get 
		$$
		Tr(L_{j})=\Sigma_{k}Tr(F_{k}S_{j})=Tr(\Sigma_{k}F_{k})S_{j}
		$$
		But $(\Sigma_{k}F_{k}=I$ and $Tr(S_{j})=1$ ,since $\{F_{k}:k=1,2,..,n\}$ is positive operator valued measure (POVM) and $S_{j}\& R_{k}$ are density matrices..Therefore we have 
		$$
		Tr(S_{j})=1.
		$$    
		Thus we have $L_{j}$  is density for each $j$. Hence $\Phi_{1}\circ \Phi_{2}$ has the following  Holevo representation
		\begin{equation}\label{eq18}
			\Phi_{1}\circ \Phi_{2}(\rho)=\Sigma_{j}Tr(H_{j}\rho)L_{j.}
		\end{equation}
		This completes the proof.             
	\end{proof}	
	Each element in $HQC(M_{n})$ can be associted with a stochastic matrix and vice versa \cite{KRB}.The following  theorems examine the structure of this association.
	
	\begin{theorem}\label{thm6}
		Assume that  $\Phi$ has the Holevo form (6.1) and let the associated  stochastic matrix $A(\Phi)$ with $(i,j)^{th}$ entry $a_{i,j}$ be as follows:
		\begin{equation}\label{eq19}
			a_{i,j}=Tr(F_{j}R_{k})
		\end{equation}
		$i,j=1,2,..,r.$.Then the quantum channel (6.1) in the Holevo form is an idempotent if and only if the probability vectors $(Tr(\rho F_{1}), Tr(\rho F_{2}),..., Tr(\rho F_{r}))^{t}$ are all steady-state vectors of the corresponding stochastic matrix (6.7) \cite{JKR}.
	\end{theorem}
	\begin{proof}
		It is known that the above matrix $A((\Phi))$ is a stochastic matrix \cite{JKR}.Now,$\Phi^{2}=\Phi$ implies that
		\begin{equation}\label{eq20}
			\Sigma_{i=1}^{n}	Tr(\Phi(\rho)F_{i})R_{i}=\Phi(\rho)
		\end{equation}
		Therefore we have,
		\begin{equation}\label{eq21}
			\Sigma_{i=1}^{n}	Tr(\Sigma_{j=1}^{r}Tr(\rho F_{j})R_{j}F_{i})R_{i}=\Sigma_{k=1}^{r}Tr(\rho F_{k})R_{k}.
		\end{equation}
		i.e., we have the following identity 
		\begin{equation}\label{eq22}
			\Sigma_{k=1}^{r}Tr(\rho F_{k})R_{k}.	=\Sigma_{i=1}^{r} [\Sigma_{j=1}^{r}Tr{(\rho F_{j})Tr{(R_{j}F_{i})}}]R_{i}
		\end{equation}
		If $\{R_{k}:k=1,2,...,r\}$ is a linearly independent set, the above equations to a system of linear equations:
		\begin{equation}\label{eq23}
			[\Sigma_{j=1}^{r}Tr{(\rho F_{j})Tr{(R_{j}F_{k})}}] = Tr(\rho F_{k})
		\end{equation}
		for $k=1,2,...,r$, 
		for all density matrices $\rho$.\\
		We can rewrite the above equation as a matrix equation as
		
		\begin{equation}\label{eq24}
			\begin{bmatrix}
				a_{1,1} \quad a_{1,2}\quad ,...a_{1,r}\\
				a_{2,1} \quad a_{2,2},\quad ...,        a_{2,r}\\
				.\quad.,...,                   . \\
				. \quad ,...,.              \\
				a_{r,1} \quad a_{r,2} \quad ,\quad ..... a_{r,r}.
			\end{bmatrix}
			\begin{bmatrix} Tr(\rho F_{1})\\
				Tr{(\rho F_{2})}\\
				.\\
				.\\
				.\\
				Tr(\rho F_{r}).
			\end{bmatrix}
			=\begin{bmatrix}Tr(\rho F_{1})\\
				Tr{(\rho F_{2})}\\
				.\\
				.\\
				.\\
				Tr(\rho F_{r}).
			\end{bmatrix}.
		\end{equation}
	\end{proof}
\begin{corollory}
		Assume that  $\Phi$ has the Holevo form (6.1) and let the associated  stochastic matrix $A(\Phi)$ with $(i,j)^{th}$ entry $a_{i,j}$ be as follows:
	\begin{equation}\label{eq19}
		a_{i,j}=Tr(F_{j}R_{k})
	\end{equation}
	$i,j=1,2,..,r.$.Then the quantum channel (6.1) in the Holevo form is an idempotent if and only if the corresponding stochastic matrix (6.7) \cite{JKR} is an idempotent.
	\end{corollory}
	\begin{proof}
		The proof follows immediately by taking the density matrix $\rho= R_{k}$ for $\{k=1,2,...,r\}$.
		\end{proof}
	\begin{remark}\label{rem3}
		It is important to note that the well known charecterisation of stochastic idempotent matrices can be used to check the condition in the above Corollory.
		It is also trivial to see that when $\Phi=P_{U}$, the matrix (19) above is the identity matrix, and the equation \ref{eq24} is trivially satisfied.\\
		Another related question is the generalized invertibility of quantum channels. However, Theorem 6.3 above will be useful to address the question of quantum channels in the Holevo form. Recall the definition of the generalized inverse of an element in a semigroup.
	\end{remark}
	\begin{remark}\label{rem4}
		By Theorem 3, section 2 \cite{HOR},$\Phi \in HQC(M_{n})$ if and only if $\Phi$ has a Kraus representation 
		\begin{equation}\label{eq25}
			\Phi(T)=\Sigma_{k=1}^{d} V_{k}^{*} TV_{k},
		\end{equation}
		where each $V_{k}$ is of rank one and  for all $T$. A simple, direct computation shows that the composition of any two maps in  has Holevo form.However, the details are provided below since explicit expression of this will be of use at a later context.\\
		Let $\Phi_{j}:j=1,2$ be in $HQC(M_{n})$ . Without loss of generality,we may assume that 
		\begin{equation}\label{eq26}
			\Phi_{j}(T)=\Sigma_{k}^{m_{j}} V_{k,j}^{*} T V_{k,j}
		\end{equation} 
		where $V_{k,j}$ has rank one for all $k$ and $j=1,2.$, for all $T\in M_{n}(C)$. Therefore we have
		\begin{equation}\label{eq27}
			\Phi_{1}(\Phi_{2}(T))=\Sigma_{k=1}^{m_{1}}\Sigma_{l=1}^{m_{2}}V_{k,1}^{*} V_{l,2}^{*} TV_{l,2}V_{k,1},
		\end{equation}
		and it is clear that if $V_{k,1}$ is not orthogonal to $V_{l,2}$ then the product $V_{k,1}V_{l,2}$ is of rank one.This completes the wanted cmputation.
	\end{remark}
	\begin{definition}\label{def7}
		Let $\mathcal{S}$ be a semigroup and $a\in \mathcal{S}$.An element $a^{\dagger}$ is called a generalised inverse of $a$ if 
		\begin{equation}\label{eq28}
			a a^{\dagger} a =a.
		\end{equation}
		It is called a semi-inverse if in addition 
		\begin{equation}\label{eq29}
			a^{\dagger} a a^{\dagger} =a^{\dagger}.
		\end{equation}
	\end{definition}
	Now the equations (28) and (29) above implies that $a a^{\dagger}$ and $a^{\dagger} a$ are idempotents.If $a^{\dagger} \in HQC(M_{n})$ ,then $aa^{\dagger }$ is an idempotent in $HQC(M_{n})$.\\
	We investigate the relationship between the compositon of elements in $HQC(M_{n})$ and the product of the corresponding stochastic matrices.  
	\begin{remark}\label{rem5}
		By definition, the associated stochastic matrix of $\Phi_{1}\circ \Phi_{2}$ is the matrix
		$a_{i,j}$ where 
		\begin{equation}\label{eq30}
			a_{i,j}=Tr(H_{i}L_{j})=Tr(H_{i}L_{j}) =Tr(H_{i}\Sigma_{k}Tr(F_{k}S_{j})R_{k}),
		\end{equation}
		$$
		=\Sigma_{k}Tr(F_{k}S_{j})Tr(R_{k}H_{i}).
		$$
	\end{remark}
	Next, we use equation \ref{eq30} above to relate the associated stochastic matrices of $\Phi_{1},\Phi_{2}\quad \& \quad \Phi_{1}\circ \Phi_{2}$.	The summary is as follows. For each quantum channel $\Phi$ in the Holevo form, let $[\Phi]$ denote the corresponding Stochastic matrix. Then we have the following theorem.
	\begin{theorem}\label{thm7}
		Let $\Phi\quad\&\quad\Psi$ be quantum channels in the Holevo form. Then
		\begin{equation}\label{eq31}
			[\Psi \circ \Phi] = [\Psi] [\Phi]
		\end{equation}
		where the modified matrix multiplication is as in equation \ref{eq30}.
	\end{theorem}
	\begin{proof}
		By equation \ref{eq30}, the associated stochastic matrix of $\Phi_{1}\circ \Phi_{2}$ is the matrix
		$a_{i,j}$ where 
		\begin{equation}\label{eq32}
			a_{i,j}=Tr(H_{i}L_{j})=Tr(H_{i}L_{j}) =Tr(H_{i}\Sigma_{k}Tr(F_{k}S_{j})R_{k}),
		\end{equation}
		$$
		=\Sigma_{k}Tr(F_{k}S_{j})Tr(R_{k}H_{i}).
		$$
		Next, we use equation \ref{eq30} above to relate the associated stochastic matrices of $\Phi_{1},\Phi_{2}\quad \& \quad \Phi_{1}\circ \Phi_{2}$.
	\end{proof}
	\begin{remark}\label{rem6}
		By using the well known  charecterisation of idempotent stochastic matrices by J.L.Doob \cite{JLD} ,the above theorem will provide a clear picture of idempotent quantum channels in the Holevo-form.One can also get information about generalised invertibility by using Theorem 7 above and the  charecterisation of generalised invertibile stocahstic matrices,Theorem 2,p.152 ,J.R.Wall \cite{JRW}.
	\end{remark}
	\begin{definition}{Entnaglement fidelity}\label{def8}
		
		Let $S=	 \vert \psi_{AR}\rangle \langle \Psi_{AR} \vert  $ be a pure state in the Hilbertspace $\mathcal{H_{A}}\otimes \mathcal{H_{R}}$.The \textbf{entanglement fidelity}
		$F_{e}(S ,\Phi)$ is defined as follows:	
		$$
		F_{e}(S ,\Phi)=\langle\psi_{AR}\rangle \vert (\Psi\otimes  Id_{R})( \vert \psi_{AR}\rangle \langle \psi_{AR} \vert ) \vert \Psi_{AR}\rangle,
		$$
		If the channel $\Psi(.)$ has the following Kraus representation
		$$
		\Psi(S)= \Sigma_{k=1}^{d}V_{K}SV_{K}^{*},
		$$
		then 
		$$
		F_{e}(S,\Psi)=\Sigma_{k=1}^{d}\vert T_{r}V_{K}S\vert ^{2}.
		$$
	\end{definition}	
	\begin{proposition}
		$P_{U}$ is a classical -quantum-classical (c-q-c) channel: equivalently, it is an entanglement-breaking- channel.
	\end{proposition}
	\begin{proof}
		Recall that $P_U{X}=\Sigma_{j=1}^{n}p_{j}Xp_{j}$  where $\{p_{j}: j=1,2,...,n\}$ are pairwise orthogonal rank one 
		projections with $\Sigma_{j=1}^{n}p_{j}=I$.But $Xp_{j}=\lambda_{j}(X)p_{j}$ ,$\lambda_{j}(X)$ are scalars since $p_{j}$ is minimal for each $j$.By taking trace we find that 
		$$
		\lambda_{j}(X)= T_{r}(Xp_{j}),
		$$
		for each $X$ and $ j$.
		Hence we get,for $X\in M_{n}$
		$$
		P_{U}(X)=\Sigma_{j=1}^{n}p_{j} T_{R}(Xp_{j}).
		$$
		Therefore, $P_{U}$ is a C-Q channel, hence entanglement-breaking.
	\end{proof}
	Next, we examine the type of channels $\tilde{P}_{U}$, in general.
	\begin{theorem}\label{thm8}
		Among  quantum channels $\tilde{P}_{U}$, only $P_{U}$ is a c-q-c channel.
	\end{theorem}
	\begin{proof}
		Let if possible $\tilde{P_{U}}$ has Holevo form,
		\begin{equation}\label{eq33}
			\tilde{P}_{U}(\rho) =\Sigma P_{k}\rho P_{k}=\Sigma_{k=1}^{r} Tr(\rho F_{k})R_{k},
		\end{equation}
		where $F_{k}$ and $R_{k}$ are POVM (postive operator valued measures) and density operators respectively and $\{P_{k}:k=1,2,...,m\}$ is a collection of pairwise orthogonal projections such that $\Sigma_{k=1}^{m}P_{k}=I$ and $m<n$,for every density $\rho$.Let $P_{k_{0}}$ be such that rank $P_{k_{0}}$ > 1.This is because $	\tilde{P}_{U}$ is differenr from $P_{U}$.
		Consider a density $\rho=x\otimes x$ $P_{k_{0}} x=x, \& \|x\| =1$ where $k_{0}$ is arbitrarily chosen.Then we have
		\begin{equation}	\label{eq34}
			\Sigma_{k=1}^{m} P_{k} x\otimes x P_{k}=\Sigma_{k=1}^{r} Tr(x\otimes x F_{k})R_{k}.
		\end{equation}
		A straightforward computation shows that 
		\begin{equation}\label{eq35}
			\Sigma_{k=1}^{m} P_{k} x\otimes x P_{k}=x\otimes x.
		\end{equation}
	\end{proof}
	\begin{remark}\label{rem7}
		Algebraically, the problem of identifying idempotents and identifying invertible/generalized invertible elements is important. Notice that there are no invertible elements in $HQC(M_{n})$.
	\end{remark}
	%\section{\textbf{Choi-Matrix of $\tilde{P}_{U}$}}
	%Next we compute the $\textbf{Choi-matrix}$ of $\tilde{P}_{U}$ for various choices of the %unitary matrix $U$.\\Recall that for a positive linear map $\Phi$ on $M_{n}$,the corresponding Choi-matrix is defined as 
	%\begin{equation}
	%J(\Phi) = \Sigma_{i,j=1}^{n} E_{i,j}\otimes \Phi(E_{i,j})
	%\end{equation}
	%where $E_{i,j}(k,l)=1,i=k \& j=l , \&  0 $ otherwise.
	\section{Codes ,Achievable rates and Channel Capacity}\label{sec4code}
	In this  section, the channel capacity of $P_{U}$ is found to be $\infty$.  In what follows, we rely heavily on the concepts and results given in p.p. $147$, chapter $8$ \cite{ASH}. Let us recall the definition of \textbf{the classical capacity of the quantum channel} \cite{ASH} and a few preliminaries and notations.
	We rstrict to finite dimentional case here.Let $H=H_{A}=H_{B}= C^{n }$ and $\Phi:H \rightarrow H$ be a Quantum Channel and 
	$$
	\Phi^{\bigotimes{n}}:B(H^{\bigotimes{n}}) \rightarrow B(H^{\bigotimes{n}}), composite  ,  memoryless- channel.
	$$
	\textbf{Notations}:$\{i=1,2,...M\}$ denotes classical messages and correspondingly $\{S_{i}:i=1,2,...M\}$ denotes density operators and $i \rightarrow S_{i}$ encodes each $i$.Now $\Phi(S_{i}) $ decodes.Correspondingly the block code for the composite channel comprises a c-q channel $i \rightarrow S_{i}^{(n)}$ that encodes the classical messages $i$ into input states $S_{i}^{(n)}$ in the space $H^{\bigotimes{n}}$, and a q-c channel observable $M^{(n)}$, here in the same space, decoding the output states $\Phi^{\bigotimes{n}[S_{i}^{(n)}]}$ into classical messages $j$.The following is a diagrammatic representation of this :
	$$
	i \rightarrow S_{i}^{(n)} \rightarrow \Phi^{\bigotimes {n}}[S_{i}^{(n)}]\rightarrow {j}.
	$$
	\begin{definition}\label{def9}
		A code $(\Sigma^{(n)}, M^{(n)})$ of length $n$ and size $N$ for the composite channel $\Phi^{\bigotimes {n}}$ consists of an encoding, given by a collection of states 
		$$
		\Sigma^{(n)}= \{S_{i}^{(n)}; i=0,1,2,...,J\}
		$$
		in $H^{\bigotimes}{n}$,and a decoding ,described by an obserable 
		$$
		M^{(n)}= \{M_{J}^{(n)}; j=1,2,...,J.\}
		$$
		in $H^{(n)}$.
	\end{definition}
	\begin{definition}{\textbf{The Classical Capacity of a Quantum Channel}}\label{def10}\\
		The classical capacity $C(\Phi)$of the quantum channel $\Phi$ is defined as the least upper bound of the achievable rates $R$ such that 
		$$
		lim_{n \rightarrow \infty} p_{e}(n  2^{nR})=0
		$$
		where, for given $n$ and size $J$ 
		$$
		p_{e}(n ,N)= min\{p_{e}(\Sigma^{(n)} , M^{(n)})\}=max_{i=1,2,...,N}[1-P_{\Sigma M}(i.i)].	
		$$
		Here for each pair$(j,i)$the probability that the outcome is $i$ when the input is $j$ is given by 
		\begin{equation}
			P_{M}(j,i)=Tr\Phi^{\bigotimes}(n)[S_{i}^{n}]M^{\bigotimes(n)}_{j}.
			\end{equation}
	\end{definition}
	\begin{theorem}\label{thm9}
		Let	$P_{U}(.)$ be as above.Then the classical capacity of the classical-quantum-channel 	$P_{U}(.)$ has finite capacity.i.e,
		\begin{equation}\label{eq36}
			C(P_{U} (.)) = Log_{2}(J).
		\end{equation}
	\end{theorem}
	\begin{proof}
		Now consider the pre-conditioner $P_{U} (.)$ where $S_{j} = p_{j} $ for each$j$ and 
		\begin{equation}\label{eq37}
			P_{U}(X) = \Sigma_{j=1}^{N}p_{j}T_{r}(Xp_{j})
		\end{equation}
		Hence a code $(\Sigma^{(n)},M^{(n)})$ of length $n$ and size $N$ for the composite channel $P^{\bigotimes n}_{U}$ is given by 
		\begin{equation}\label{eq38}
			\Sigma_{(n)} =\{ p^{(n)}_{j} : j = 1,2,..J \}  \&  M^{(n)} = \{ p_{j}^{(n)} : j = 0,1,2,...,J.\}	
		\end{equation}
		It is more or less clear that $p_{\Sigma M}(i\vert i)=1$ for all values of $i$ and hence $p_{e}(n,J) = 0 $ for all vaues of $n$ and $J$.Thus with the above mentioned code $P_{U}(.)$ has capasity $Log_{2}(J)$,by defenition 10 above.
		
		\begin{equation}\label{eq39}
			C(P_{U} (.)) = Log_{2}(J).
		\end{equation}
	\end{proof}
\begin{remark}\label{rem7}
	Equation (40) above shows that the capacity $	C(P_{U} (.))\longrightarrow \infty $ logarithemically as the size $J\longrightarrow \infty$.In Example1 ,each $U_{n} or S_{n}$ depends on the number of grid points chosen.Thus we infer that the capacity keeps increasing logarethemically as the number of grid points increases.
	\end{remark}
	\begin{remark}\label{rem8}
		Before  considering next aspect, we recall coding-encoding techniques discribed by Kamil Korzekwa, Zbigniew Puchała,etal.,in \cite{KZMK} and Zi-Wen Liu, Xueyuan Hu,etal.,in \cite{ZXS} using idempotent quantum-channels,the so called \textbf{ resource destroying channels }.The following basic identities involving idempotent channels are used in the above articles;
		\begin{equation}\label{eq40}
			\Phi \circ \Delta=\Delta=\Delta\circ \Phi
		\end{equation}
		where $\Delta$	is an idempotent resource destroying channel and $\Psi$ and is an eppropriate quantum channel that encodes\cite{KZMK}.\\
		A variant of equation \ref{eq36} above was considered in \cite{ZXS} Iis as follows;
		\begin{equation}\label{eq41}
			\Phi\circ \Delta=\Delta \circ \Phi \circ \Delta,
		\end{equation}
		and its dual form
		\begin{equation}\label{eq42}
			\Delta\circ \Phi=\Delta \circ \Phi \circ \Delta.
		\end{equation}
		A choice of $\Delta$ can be the preconditioner map $P_{U}$ and $\Phi=\tilde{P}_{U}$ for various unitaries $U$.\\
		\begin{example}
			Another example could be as follows.Let $S$ be a permution of $\{1,2,...,n\}$ where $n$ is the dimention appearinf in $M_{n}(C)$.Define $\Phi_{S}$ as follows;
			\begin{equation}\label{eq43}
				\Phi_{S}(T)=US\circ\sigma(U^*TU) U^{*},
			\end{equation}
			$T\in M_{n}(C)$.Then $\Phi_{S}$ is a quantum channel for each $S$.
		\end{example}
		It can be easily seen that $\Delta$ and $\phi_{S}$ satisfy the following equations.
		
		\begin{itemize}
			
			\item[{[1]}]
			$$
			\Phi\circ\Phi_{S}=\Phi_{S}= \Phi_{S}\circ \Phi,and \quad therefore
			$$
			\item[{[2]}]
			$$
			\Phi\circ\Phi_{S}\circ\Phi=\Phi_{S}.
			$$		
		\end{itemize}
		
		Observe that $\Phi_{S}$ is not an idempotent  quantum-channel, unless of course $S^{2}= I$.
	\end{remark}
	
	\section{Infinite Dimensional Analogue of the Maps $\tilde{P}_{U}$}\label{sec5infi}
	In what follows, we consider the infinite-dimensional versions of the results presented in the preceding sections.
	
	Let $H$ be an infinite-dimensional separable Hilbert space over complex numbers, and $\tau (H)$ be the Banach $*$-algebra of all trace-class operators equipped with the trace norm. For $A \in \tau(H)$, the trace norm of $A$ is defined as:
	\begin{equation}\label{eq44}
		|A|_{1}= \text{Tr}(|A|),
	\end{equation}
	where $|A|=\sqrt{(A^{}A)}$.
	
	Quantum information theory and quantum channels in infinite dimensions are mathematically interesting and have recently caught the attention of applied signal analysts. The survey article \cite{VSS} provides useful information in this direction. Additionally, the monograph by Brian Davies \cite{EBD} is also highly important. However, it is clear that one can mimic the finite-dimensional theory and appropriately extend it to infinite dimensions. First, we recall the definition of a quantum channel in infinite dimension.
	\begin{definition}\cite{HOS}.
		Let $H$ be a separable complex Hilbert space, and $\tau(H)$ be the Banach $*$-algebra of all trace-class operators on $H$. Then, a quantum channel is defined as a linear, trace-preserving map $\Phi$ on $\tau(H)$ such that the dual map $\Phi^{*}$ on the $C^*$-algebra $B(H)$ is completely positive.
		\end{definition}
	
\begin{remark}
	Every quantum channel $\Phi$ is continuous under the trace norm. Since its dual map $\Phi^{*}$ is completely positive, it is continuous under the operator norm on $B(H)$. Moreover, every linear map on a Banach space with a continuous dual is continuous, which can be shown using a simple application of the Hahn-Banach theorem.
	
	Furthermore, the quantum channel $\Phi$, being continuous under the trace norm, implies that its dual $\Phi^{*}$ is an ultra-weakly continuous (i.e., normal) completely positive map on $B(H)$. Therefore, there exists a sequence ${V_k: k = 1, 2, \ldots}$ in $B(H)$ such that \cite{EBD}, \cite{SF}:
	\begin{equation}\label{eq45}
		\Phi^{*}(T) = \sum_{k=1}^{\infty} V_k^{*}TV_k,
	\end{equation}
	and hence its pre-dual map is:
	\begin{equation}\label{eq46}
		\Phi^{*}(S) = \sum_{j=1}^\infty V_kSV_k^{*},
	\end{equation}
	for all $S \in \tau(H)$.
	
	Next, we introduce the infinite-dimensional versions of the preconditioner maps $P_U$ and $\tilde{P}U$ as follows. We mimic the existing definition of $P_U$ as well as its general version $\tilde{P}U$. Let $U$ be a unitary operator in $B(H)$, and let ${P_j: j = 1, 2, \ldots}$ be a set of rank-one orthogonal projections such that $\sum{j=1}^\infty P_j = I$, where $I$ is the identity operator on $H$. For each $S \in \tau(H)$, define $\sigma(S) = \sum{j=1}^\infty P_jSP_j$. Now, as before, we define a closed subspace $M(U)$ as follows:
	\begin{equation}\label{eq47}
		M(U) = {T \in \tau(H): U^{*}TU = \sigma(S) \text{ for some } S \in \tau(H)}.
	\end{equation}
	
	The continuity of the map $\sigma$ under the norm $|.|_1$ implies that $M(U)$ is a closed subspace of the Hilbert space $\tau(H)$. We will now prove the infinite-dimensional version of the preconditioner map $P_U$.
\end{remark}
	\begin{theorem}
		Let $M(U)$ be the closed subspace of $\tau(H)$ given by equation \ref{eq47} above, and let $P$ be the orthogonal projection of $\tau(H)$ onto $M(U)$. Then $P$ has the following properties:
		
		\begin{enumerate}
			\item $P(T) = U\sigma(U^TU)U^{*}$ for $T \in \tau(H)$, where $\sigma(T) = \sum_{j=1}^{\infty}P_jTP_j$ for each $T \in \tau(H)$.
			\item $P$ is completely positive and trace preserving.
			\item For each $T \in \tau(H)$, $\|T-P(T)\|_{1}^{2} = \|T\|_{1}^{2} - \|P(T)\|_{1}^{2}$.
		\end{enumerate}
	\end{theorem}
\begin{proof}
	\begin{enumerate}
		\item Clearly, $P(T) \in M(U)$ for all $T \in \tau(H)$, and the range of $P$ is $M(U)$. We show that $P$ is self-adjoint and idempotent, which completes the proof by uniqueness.
			We have:
			\begin{align*}
				P(P(T)) &= U\sigma(U^*P(T)U)U^* \\
				&= U\sigma(U^*U\sigma(U^*TU)U^*U)U^* \\
				&= U\sigma(U^*TU)U^* \\
				&= P(T).
			\end{align*}
			Thus, $P$ is idempotent.
			
			Next, we show that $P$ is self-adjoint. For $S, T \in \tau(H)$, consider the following:
			\begin{align*}
				\langle P(T), S \rangle &= \operatorname{Tr}(P(T)S^*) \\
				&= \operatorname{Tr}(U\Sigma_{j=1}^{\infty}P_{j}U^*TUP_{J}U^*S^*) \\
				&= \operatorname{Tr}(S^*U\Sigma_{j=1}^{\infty}P_{j}U^*TUP_{j}U^*) \\
				&= \operatorname{Tr}(U^*S^*U)\Sigma_{j=1}^{\infty}P_{j}U^*TUP_{j} \\
				&= \Sigma_{j=1}^{\infty}\operatorname{Tr}(P_{j}U^*S^*UP_{j}U^*TU) \\
				&= \operatorname{Tr}(TU\Sigma_{j=1}^{\infty}P_{j}U^*S^*UP_{j}) \\
				&= \operatorname{Tr}(TP(S)^*) \\
				&= \langle T, P(S) \rangle.
			\end{align*}
			This shows that $P$ is self-adjoint. Hence, $P$ is the orthogonal projection of $\tau(H)$ onto $M(U)$.
			
			\item It is clear that the map $P$ has the following Kraus representation:
			\[
			P(T) = \sum_{j=1}^{\infty} U P_{j}U^* T UP_{j}U^*.
			\]
			Therefore, $P$ is completely positive.
			
			\item It follows from the Pythagorean theorem in abstract Hilbert spaces.
			\end{enumerate}
		\end{proof}
\begin{remark}
	We can define the infinite-dimensional version of $\tilde{P}_{U}$ as follows:
	
	\begin{definition}
		Let ${p_{k}: k = 1, 2, \ldots}$ be a sequence of orthogonal and pairwise orthogonal projections such that $\sum_{k=1}^{\infty} p_{k} = I$. Let
		\begin{equation}\label{eq48}
			\Psi(T) = \sum_{k=1}^{\infty} p_{k}Tp_{k}.
		\end{equation}
		For a unitary $U \in B(H)$, the map $\tilde{P}_{U}$ is defined as
\[
\tilde{P}_{U}(T)=U\Psi(U^*TU) U^*.
\]
for each $T \in \tau(H)$.
\end{definition}
\end{remark}
There are a number of concrete unitary operators in infinite dimention such as Fourier transform on $L^2(R^N)$.One can also create more examplles from finite dimention as follows:\\
Let $H_{n}$ be an $n$-dimentional subspace of a seperable Hilbert space $H$ and let $U_{n}$ be a unitary operator on $H_{n}$.Define a unitary $U$ on $H$ as follows:
\begin{equation}
	U^{(n)}=U_{n}\oplus (I-P_{n}),
	\end{equation}
where $P_{n}$ is the orthogonalprojection of $H$ onto $H_{n}$.
\section{Infinite-Dimentional Holevo Forms}

\begin{remark}
	The extended unitaries $U^{(n)}$ gives a sequence of unital completely positive maps $\{P_{U^{(n)}}\}$ on $B(H)$. This bounded sequence will provide cluster points in the BW-topology of Arveson.
	\end{remark}

We conclude this article by noting that our future work will include the infinite-dimensional version of the results presented in this article.In addition to the concrete unitaries listed at the beginning,one can also consider the Discrete Quantum Fourier Transform ,which may be more suitable for the theme of this article.However,it is yet to analyse this aspect.

\section*{Acknowledgements:} I am thankful to KSCSTE, Government of Kerala, for the financial support provided through the Emeritus Scientist Fellowship, which enabled the completion of a part of this article. This work was also presented at the 'XI International Conference of the Georgian Mathematical Union, Aug. 2021, Batumi, Georgia'. Additionally, I would like to express my gratitude to the Institute of Mathematics, Research and Training (IMRT), Thiruvananthapuram, for the weekly research lectures on topics related to the theory of semigroups.		
		\nocite{*}
		\bibliography{quantum}% Produces the bibliography via BibTeX.
	\end{document}